\documentclass[a4paper,UKenglish]{lipics}
 
\usepackage{microtype}
\usepackage{multirow}
\usepackage{colortbl}
\usepackage{amssymb}
\usepackage{graphicx}
\usepackage{tikz,tikz-qtree}

\newtheorem{proposition}{Proposition}

\title{Set theory and tableaux for teaching propositional logic}
\titlerunning{Set theory and tableaux for teaching propositional logic}

\author[1]{Nino Guallart\footnote{This work has been partially supported by a grant from the V Plan Propio de la Universidad de Sevilla.}}
\author[2]{\'Angel Nepomuceno-Fern\'andez}
\affil[1]{University of Seville\\
  Seville, Spain\\
  \texttt{nguallart@us.es}}
\affil[2]{University of Seville\\
  Seville, Spain\\
  \texttt{nepomuce@us.es}}
\authorrunning{N. Guallart and A. Nepomuceno-Fern\'andez} 

\Copyright{Nino Guallart and \'Angel Nepomuceno-Fern\'andez}


\serieslogo{logo_ttl}
\volumeinfo
  {M. Antonia {Huertas}, Jo\~ao {Marcos}, Mar\'ia {Manzano}, Sophie {Pinchinat}, \\
  Fran\c{c}ois {Schwarzentruber}}
  {5}
  {4th International Conference on Tools for Teaching Logic}
  {1}
  {1}
  {45}
\EventShortName{TTL2015}

\begin{document}

\maketitle

\begin{abstract}
In this work we suggest the use of a set-theoretical interpretation of semantic tableaux for teaching propositional logic. If the student has previous notions of basic set theory, this approach to semantical tableaux can clarify her the way semantic trees operate, linking the syntactical and semantical sides of the process. Also, it may be useful for the introduction of more advanced topics in logic, like modal logic. 
 \end{abstract}

\section{Introduction}

Logic is a discipline that is studied both in sciences as in humanities, so we
need to take into account the context in which it is being taught. Our proposal is
for teaching logic to students who have completed a course about argumentation
theory or informal logic, provided certain basic notions of logic have been
studied. However, it may be suitable for students interested in other
disciplines such as computer sciences or mathematics. 

According to the cross-cutting character of logic, instead of a basic course introduced as usual, we suggest a course on propositional logic in a way that, given a propositional formal language, its semantics will be defined, not only with truth tables but also in terms of set theory. That is to say, after the presentation of semantics by means of truth tables, another presentation of semantics should be given, by means of basic notions of set theory. On the other hand, the method of tableaux introduced in (Beth, 1955) and (Smullyan, 1968) will be taken as a first form of calculus. In both cases graphical representations will be inserted, whether that are known or not (as an auxiliar tool of verification, of course after short explanations when necessary). The corresponding educational methodology allows to tackle easily several relevant topics that can be connected and becomes a general introduction to formal methods with several advantages:

\begin{enumerate}
  \item Basic set theory, whose rudiments may be known or are being studied
  simultaneously by students of computer sciences or mathematics, are put in action, which facilitates a rigorous first
  approach to mathematical logic. For students of humanities, this allows a
  gradual access to formal treatment of known issues.
  
  \item This teaching procedure makes easier the subsequent approach (and
  understanding) to semantics for first order logics and other topics within this field, like the formation of disjunctive normal form of propositional formulas.
  
  \item Likewise, it facilitates the access to propositional modal logic
  as a natural continuation of the classical one.

  \item Some applications can be understood better with tableaux method. For students of humanities calculi could be more intuitive and ulterior philosophical and metalogical problems could be addressed better. For students of mathematics  this approach may be linked to other fields, such as graph theory. For students of computer sciences tableaux can be seen as a way of working in logic next to algorithmical methods.
 
\end{enumerate}

There are some variations in the way that a semantic tableau can be developed, we will not focus on their differences. For a good starting point on the topic we recommend (D'Agostino, 1999). We will begin with the propositional language with the set theoretic semantics. Tableaux method is described and the corresponding results are settled. Then this method will be put in teaching perspective. To finish, a conluding remarks section is presented and a basic bibliography is incorporated in the paper.

\section{Classical propositional logic}

Let $L$ be a propositional language defined from a set of atoms or
propositional variables $\mathcal{P}$. The syntax must be presented to
students by means of a BNF rule:
\[ \varphi::= p \mid \neg \varphi \mid \varphi \wedge \varphi \mid
   \varphi \vee \varphi \mid \varphi \to \varphi \]
In order to define the standard semantics, let $V$ be a function from propositions to truth values $V:{\cal P}\mapsto \{0,1\}$. The evaluation of formulas is as usual. 
      
In order to settle the semantics, after presenting the standard one by means of truth tables, an introduction to basic notions of set
theory must be made. An atom is true or false with respect to a given situation or
state of the world{\footnote{From an intuitive point of view. Despite that
they were related, we are not using these terms in a specific sense of
situation theory for example.}}, according to the studied truth tables, but we should underlike the interest in the set of situations in which the atom is true, then we can consider a set of states of the
world that constitutes the framework to interpret our language $L$. A
propositional model $M$ is now defined as a set $\mathcal{U} \neq
\varnothing$, the set of states of the world, and $v_{\mathcal{U}}$ is a
function defined from $L$ to $\wp ( \mathcal{U})$.
\begin{definition}
$M = ( \mathcal{U},
v_{\mathcal{U}})$ (though we shall write $v$ instead of $v_{\mathcal{U}}$ to
abbreviate) that accomplished the following cluases:
\begin{enumerate}
  \item If $p \in \mathcal{P}$, $v (p) \in \wp ( \mathcal{U})$ ---or, what is
  the same, $v (p) \subseteq \mathcal{U}$---  
  \item $v (\neg \varphi) = \mathcal{U} \setminus \{v (\varphi)\} =
  \overline{v (\varphi)}$
  
  \item $v (\varphi \wedge \chi) = v (\varphi) \cap v (\chi)$
  
  \item $v (\varphi \vee \chi) = v (\varphi) \cup v (\chi)$
  
  \item $v (\varphi \to \chi) = \overline{v (\varphi)} \cup v (\chi)$
\end{enumerate}

\end{definition}

Given a model $M = ( \mathcal{U}, v)$ the notion of satisfaction, in a state
of the world $s$, of a formula $\varphi$ ---in symbols $M, s\models\varphi$ is given in the following ---

\begin{definition} In keeping with the grammar of $L$, 
\begin{enumerate}
  \item $M, s \models p$ if and only if{\footnote{``iff'' from now on.}} \ $s
  \in v(p)$
  
  \item $M, s \models \neg \varphi$ iff $M, s \not{\models} \varphi$
  
  \item $M, s \models \varphi \wedge \chi$ iff $M, s \models \varphi$ and $M,
  s \models \chi$
  
  \item $M, s \models \varphi \vee \chi$ iff $M, s \models \varphi$ or $M, s
  \models \chi$
  
  \item $M, s \models \varphi \to \chi$ iff $M, s \not{\models} \varphi$ or
  $M, s \models \chi$
\end{enumerate}
\end{definition}

When a formula is satisfiable in standard sense ---there is an assignement of truth values to its atoms such that the formula is true---, then we shall say that {\it the formula is standard-satisfiable}. Then there exists a situation in wich the formula is true, so a model $M$ is definible, that is to say, it is satisfiable in the sense of the semantics in terms of set theory. The relation between both forms of satiafiability is settled in the following

\begin{theorem}
A formula $\varphi\in L$ is standard-satisfiable iff it is satisfiable in the sense of the semantics in terms of set theory.
\end{theorem}

\begin{proof}
By induction over the logical degree of $\varphi$.
\end{proof}

According to Theorem 1, a wff $\varphi$ is valid iff for all propositional models $M$, and all situation $s$, $M \models \varphi$.

\begin{proposition}
By applying definition 2, the following results can be proven:
\begin{enumerate}
  \item If $M, s\models p$, then $v(p)\neq\varnothing$
  \item $M, s \models \neg \varphi$ iff $s \in \overline{v (\varphi)}$
  
  \item $M, s \models \varphi \wedge \chi$ iff $s \in v (\varphi) \cap v
  (\chi)$
  
  \item $M, s \models \varphi \vee \chi$ iff $s \in v (\varphi) \cup v (\chi)$
  
  \item $M, s \models \varphi \to \chi$ iff $s \in \overline{v (\varphi)} \cup
  v (\chi)$
\end{enumerate}
\end{proposition}
Specific exercises must be proposed to students to practice theses
notions. Example: Is $p \wedge q \to r$ satisfied in a model $M = (
\mathcal{U}, v)$ and the state $s$ provided $(v (p) \cap v (q)) \subseteq v
(r)$? Let us see both cases:
\begin{enumerate}
  \item $s \in v (r)$, then $s \in ( \overline{v (p) \cap v (q)}) \cup v (r)$
  
  \item $s \not\in v (r)$, then $s \in \overline{v (r)} \subseteq \overline{v (p)
  \cap v (q)}$ so that $s \in ( \overline{v (p) \cap v (q)}) \cup v (r)$
\end{enumerate}
then, whatever the case may be $M, s \models p \wedge q \to r$.

If we represent subsets of states with Venn diagrams, it is very easy to explain graphically the nature of logical operators. Each circle represents the set of states of the world that make true a certain sentence, $\varphi$ or $\chi$ in this case. In each diagram the whole picture is the universe of states of the world $\mathcal{U}$ and the grey area is the set of states that make true a certain well-formed formula (wff from now on):

%
%
%
%
%

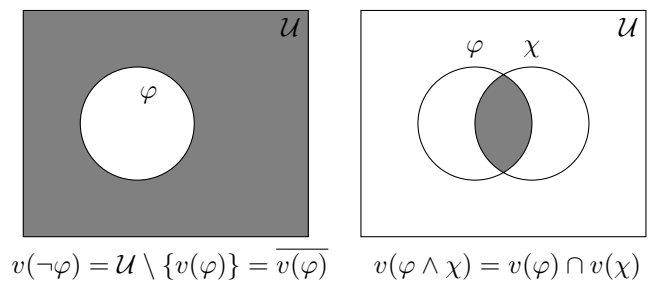
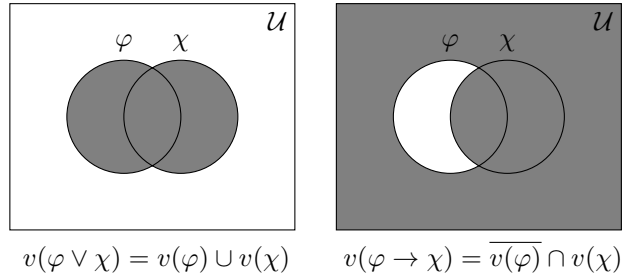
\begin{center}
\begin{tabular}{cc}

\begin{tikzpicture}[fill=gray,scale=0.75] 
\scope \clip (-2,-2) rectangle (3,2)       ; 
\fill (-2,-2) rectangle (3,2);
\fill [color=white] (0,0) circle(1); \endscope;
\draw (2.7,1.7) node {$\mathcal{U}$};
\draw (0,0) circle (1) (0.2,0.2)  node [text=black,above] {$\varphi$}   (-2,-2) rectangle (3,2) node [text=black,above] {$$}; 
\end{tikzpicture} 

 & 
\begin{tikzpicture}[fill=gray,scale=0.75] 
\scope \clip (-2,-2)   (0,0) circle (1); 
\fill (1,0) circle (1); \endscope
\draw (2.7,1.7) node {$\mathcal{U}$};
\draw (0,0) circle (1) (0,1)  node [text=black,above] {$\varphi$}       (1,0) circle (1) (1,1)  node [text=black,above] {$\chi$}       (-2,-2) rectangle (3,2) node [text=black,above] {$$}; 
\end{tikzpicture} 

\tabularnewline
$v (\neg \varphi) = \mathcal{U} \setminus \{v (\varphi)\} = \overline{v (\varphi)}$
&
$v (\varphi \land \chi) = v (\varphi) \cap v (\chi)$

\end{tabular}
\end{center}

\begin{center}
\begin{tabular}{cc}

\begin{tikzpicture}[fill=gray,scale=0.75] 
\draw (2.7,1.7) node {$\mathcal{U}$};
\fill (0,0) circle (1);
\fill (1,0) circle (1); 
\draw (0,0) circle (1) (0,1)  node [text=black,above] {$\varphi$}       (1,0) circle (1) (1,1)  node [text=black,above] {$\chi$}       (-2,-2) rectangle (3,2) node [text=black,above] {$$}; 
\end{tikzpicture} 

 & 

\begin{tikzpicture}[fill=gray,scale=0.75] 
\fill (-2,-2) rectangle (3,2);
\fill [fill=white](0,0) circle (1);
\draw (2.7,1.7) node {$\mathcal{U}$};
\fill  (1,0) circle (1); 
\draw (0,0) circle (1) (0,1)  node [text=black,above] {$\varphi$}       (1,0) circle (1) (1,1)  node [text=black,above] {$\chi$}       (-2,-2) rectangle (3,2) node [text=black,above] {$$}; 
\end{tikzpicture} 
\tabularnewline
$v (\varphi \vee \chi) = v (\varphi) \cup v (\chi)$
&
$v (\varphi \to \chi) = \overline{v (\varphi)} \cap v (\chi)$
\end{tabular}
\par\end{center}

\section{Basic theory for semantic tableaux method}

{In this section we are going to introduce basic definitions
and theorems for the use of semantic tableaux. Since the aim of this paper is
pedagogical, not theoretical, this section has a simple purpose, namely to clarify and
justify the methods that will be applied as it is shown in the next section.
We will focus on adopting semantic tableaux method for propositional classical logic. Given a finite set of formulas $\Gamma$, $T (\Gamma)$ denotes the semantic tableau,
which is a sequence of sequences of formulas called branches. Each branch is
obtained by means of application of rules to formulas that are not literals
until the branch is closed (when two complementary literals occur into it) or
there is no complex formula without applicating the corresponding rule. A
classification of complex formulas:

\begin{center}
  {\vspace{0.5cm}}\begin{tabular}{|c|c|c|}
    \hline
    $\alpha$ & $\alpha_1$ & $\alpha_2$\\
    \hline
    $\varphi \wedge \chi$ & $\varphi$ & $\chi$\\
    \hline
    $\neg (\varphi \vee \chi)$ & $\neg \varphi$ & $\neg \chi$\\
    \hline
    $\neg (\varphi \to \chi)$ & $\varphi$ & $\neg \chi$\\
    \hline
  \end{tabular} {\hspace{0.5cm}} \ \begin{tabular}{|c|c|c|}
    \hline
    $\beta$ & $\beta_1$ & $\beta_2$\\
    \hline
    $\varphi \vee \chi$ & $\varphi$ & $\chi$\\
    \hline
    $\neg (\varphi \wedge \chi)$ & $\neg \varphi$ & $\neg \chi$\\
    \hline
    $\varphi \to \chi$ & $\neg \varphi$ & $\chi$\\
    \hline
  \end{tabular} {\vspace{0.5cm}}
\end{center}

There are three kinds of rules:
\[ \frac{\neg \neg \varphi}{\varphi} ; \hspace{1em} \frac{\alpha}{\alpha_1 ;
   \alpha_2} ; \hspace{1em} \frac{\beta}{\beta_1 \mid \beta_2} \]

The following theorems
establish 
the way the
algorithm for solving
tableaux works, giving
a finite tableay for a finite
set of formulas which is satisfiable if it is open, and the
relationship between
this method and the
formation of disjunctive
normal forms.

The fundamental property of the tableaux is given by the following theorem:

\begin{theorem}
  A finite set of formulas $\Gamma$ is satisfiable iff $T (\Gamma)$ is open.
\end{theorem}

\begin{proof}
  The proof is very easy and can be found in any text related to semantic tableaux. Here we take it for granted, given the practical nature of this work.
\end{proof}

Note that for a finite set of formulas $\Gamma$, its tableau is equivalent to the tableau of a single formula, the conjunction of them.

As a corollary, $\Gamma$ is contradictory (not satisfiable) iff $T (\Gamma)$
is closed. On the other hand, let $\Gamma$ be a finite set of formulas and
$\varphi$ a formula such that the last is logical consequence of the former,
that is to say $\Gamma \models \varphi$, then semantic tableaux method can be applied to study that. Of course, $\Gamma \models \varphi$ iff the set $\Gamma \cup
\{\neg \varphi\}$ is not satisfiable, which is equivalent to say that $T
(\Gamma \cup \{\neg \varphi\})$ is closed.

For a valild wff $\phi$, its negation $\neg \phi$ has no possible interpretation so it is contradictory. In order to test whether a wff is valid, we must check its negation; $\phi$ is valid iff the semantic tableau of $\neg \phi$ is closed.

Semantic tableaux are usually represented graphically as rooted trees. As a tree, a tableau is a graph with no cycles, so for every two nodes there is exactly one path connecting them. We assign the formula we are analyzing to the root node, and their immediate children nodes are the subformulas obtained by any of the aforementioned rules.  Since every node can be seen as the root of a subtree, we can have a complete tableau defined recursively. A node with no branches or a leaf is a terminal node, which is the node of a literal.  

The following two theorems are crucial for understanding the practical use of semantic trees:

\begin{theorem}
The semantic tableau of a finite formula is a finite tree.
\end{theorem}

\begin{proof}
A child node contains a subformula of the node obtained by the rules, so it has to be a smaller formula than the origina its predecessor. The root is a finite formula or a finite set of formulas, so this means that in every branch we will reach to a literal in a finite number of steps, and it corresponds to a leaf. 
\end{proof}

Our next step is of semantical nature and crucial in our argumentation, since we want to  link the use of semantic tableaux to a set-theoretical interpretation of formulas. We will cover this more in-depth in the next section. However, in order to understand the set-theoretical interpretation of the method, we need to know
the relationship between the interpretation of literals and the interpretation of the initial formula.

\begin{theorem}
The interpretation of a formula is a recursive function of the partial interpretations of its subformulas.
\end{theorem}

\begin{proof}
Informally, the semantic tableau cuts the formula into its constituting subformulas until their literals are obtained, so we can go the opposite direction to prove the theorem; the model of the conjunction os two subformulas is the conjunction of the models of them, and the model of the disjunction of two formulas is the disjunction of their models. A more exhaustive proof will be included into the final work.
\end{proof}

\begin{definition}
A wff is in its disjunctive normal form (DNF) if it is the disjunction of a number of conjunctive clauses, which are logical expressions formed by the conjunction of a finite number of atoms:
\end{definition}

\[ \bigvee_{i=1}^{n} \bigwedge_{j=1}^{m_{i}} a_{ij}\]

\noindent where $a_{ij}$ is an atom, that is, either a literal or the negation of a literal.

Every wff is equivalent to a DNF, and the transformation can be given in a finite number of steps, via a set of rules based on De Morgan's laws, elimination of double negation, associativity and the distributive laws of one operator over the other, having in mind that $p\rightarrow q$ is equivalent to $\neg p \lor q$ and to $\neg (p\rightarrow q)$ is equivalent to $p \land \neg q$. The DNF of a wff is not unique, since there can be equivalent DNF's of a formula \footnote {A single atom is also a clause which just one atom, and therefore expressions such as $p\lor (q\land r)$ are also DNF's. {\it Complete disjunctive normal forms} are wff's in its disjunctive normal form in which all clauses contain all the literals in the wff, either positive or negative. A wff have a unique complete disjunctive normal form, but it may be equivalent to several non-complete DNF's. For example, both $\neg p \lor q$ and $(\neg p \land q)$ are equivalent to $(\neg p \land \neg q) \lor (p \land q) \lor (\neg p \land q)$. }. (See  (Quine, 1952) for more details about normal forms and the reduction of formulas to them).

\begin{remark}
{\it Semantic tableaux and disjuctive normal forms}. The method of semantic tableaux for a wff is akin to the conversion of this formula to its DNF, since the $\alpha$ and $\beta$ rules are equivalent to the rules for the formation of the DNF of a wff. There are some methodological differences, for example that the semantic tableaux gives eventually a set of subformulas, not an explicit disjunction of them, but implicitly the semantic tableau can be seen as a method for getting the DNF of a formula via the obtention of a set of satisfiable subformulas. In this way, the teaching of semantic tableaux is also a graphical method for the obtention of DNF's, maybe clearer for the students than other methods based on purely syntactical rules.

We have established that the process of the tableau is finite and, when the original formula is reduced to a set of disjunctive subformulas that are satisfiable under the same interpretation, that set satisfies the wff. Each branch of the tree gives the conjunction of some atoms, that is, of some literals of the wff and, if the branch is open, it has a valid model that satisfies it. If there is a model that satisfies a subformula, it also satisfies partially the original formula. Since this model is a function from the models of its subformulas, the semantic tableau is basically a method for obtaining its whole model; the whole tree is comprised of branches, and the disjunction of the open branches gives the smallest satisfiable model for the wff. 

\end{remark}
\section{Semantic perspective in teaching}
Until now, we have seen the theoretical preliminaries of a set-theoretical interpretation of semantic tableaux. Now we are going to show how it works in practical terms, and the usefulness of this method for teaching logic.
We are going to sketch the advantages (and some disadvantages) of set theoretical interpretation via semantic tableaux, focusing on a comparison between semantic tableaux and truth tables for determining some properties of formulas, more exactly satisfiability and validity.

\subsection{Satisfiability}
Semantic tableaux are an easy and quick method for determining the satisfiability of formulas, and their set-theoretical interpretation can be linked to the process in a simple and illustrative way. Let's take a formula in order to check whether it is satisfiable or not; for the purposes of this exposition we will make $\phi \equiv (p\lor q)\land(\neg p\lor r)$. From the initial formula, the process of semantic tableaux gives out a set of branches. Now we are going to assign a possible state to every branch, both open or closed.
 Let's recall that a formula is valid iff it has at least a possible interpretation. Since each branch of a semantic tableau corresponds to an state of a propositional model, the tree of a satisfiable wff has at least one open branch:

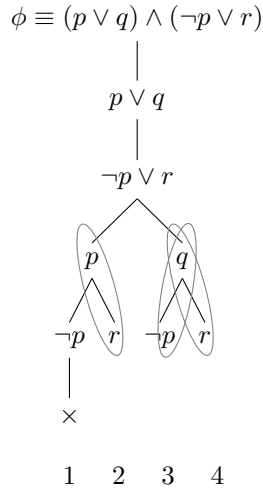
\hbox to \hsize{\hfil{
\begin{tikzpicture} 
\Tree [.$\phi \equiv (p\lor q)\land(\neg p\lor r)$ [.$p\lor q$ [.$\neg p\lor r$ [.\node(p){$p$}; [.\node(np1){$\neg p$}; [.\node(x){$\times$}; ] ] [.\node(r1){$r$}; ] ] [.\node(q){$q$}; [.\node(np){$\neg p$}; ] [.\node(r2){$r$}; ] ] ] ] ] ]
\draw [color=gray,rotate=105]  (p) +(110:0.3) +(-0.4,0) ellipse (0.9 and 0.2);
\draw [color=gray,rotate=80]  (q) +(80:0.3) +(-0.4,0) ellipse (0.9 and 0.2);
\draw [color=gray,rotate=105]  (q) +(110:0.3) +(-0.4,0) ellipse (0.9 and 0.2);

\node (1) [below of = np1,node distance=1.8cm] {1};
\node (2) [right of = 1,node distance=0.65cm] {2};
\node (3) [right of = 2,node distance=0.65cm] {3};
\node (4) [right of = 3,node distance=0.65cm] {4};

\end{tikzpicture}
}\hfil}

Observe that the steps in this tree have a set-theoretical interpretation via the rules defined the section 2:
\begin{center}
$ v((p\lor q)\land(\neg p\lor r))= v(p\lor q) \cap v(\neg p\lor r) = (v(p) \cup v(q)) \cap (v(\neg p) \cap v(r)))= $
$  (v(p)\cap \overline{v(p)}) \cup (v(p)\cap v(r)) \cup (\overline{v(p)} \cap v(r)) \cup (v(q) \cap (v(r)) = $
$ \varnothing \cup v(p\land r) \cup v(\neg p \land q) \cup v(q\land r) =$
$  v(p\land r) \cup v(\neg p \land q) \cup v(q\land r) $
\end{center}

Note that each step in this interpretation corresponds to each line of the semantic tableau, being more clearer on the latter. The set-theoretical interpretation of each branch is the set of states satisfying the formula, which is a subset of the states that satisfy the whole wff; each branch determines a state, and each open branch determines a state that satisfies the formula, so for each state $s_{i},i=2,3,4$, we have that $M,s_{i}\models \phi$. The interpretation of the wff is given by the disjunction of the sets that satisfy the branches. 

Since we have four branches, three of them open, we can assign a state to each branch, and satisfiable states are the corresponding ones to open branches, that is, all but 1. Let $S=$\{1,2,3,4\} the set of states determined by the tableau, from left to right. Clearly, $v(p)=\{2\},v(q)=\{3,4\},v(r)=\{2,4\}$, and the interpretation $v(\phi )=\{2,3,4\}=S^{*}$, that is, the union of the states of the open branches. By means of $v(p) \cap v(r)=\{2\}$, $v(q) \cap v(r)=\{3\}$ and $v(\neg p) \cap v(q)=\{4\}$ it can be seen that this model satisfies $\phi\equiv (p\lor q)\land(\neg p\lor r)$, and we can take it as a representative of the equivalence class $[S^{*}]_{\models}$ of all models that satisfy $\phi$\footnote {Actually, $S^{*}$ is the representative of the minimum model, since for any model $M$ satisfying this formula (obviously with a cardinality equal or higher to the cardinality of $S^{*}$), there is an epimorphism $M\rightarrow S^{*}$ preserving the satisfaction of the wff.}.

Graphically, we can see the same results in the form of Venn diagrams:

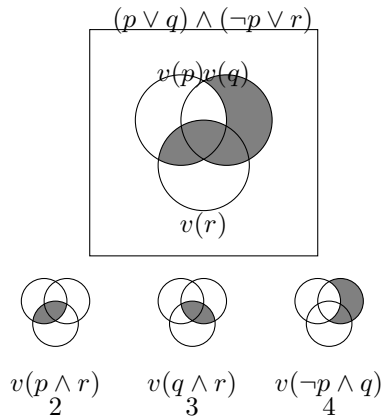
\hbox to \hsize{\hfil{
\begin{tikzpicture} [scale=0.6]
\begin{scope}[xshift=5cm] 
\draw (-2,-3) rectangle (3,2);
\draw (0.7,2.2) node {$(p\lor q)\land(\neg p\lor r)$};
\fill [gray] (1,0) circle(1);
\fill [white] (0,0) circle(1);
\scope \clip (-2,-2)   (0.5,-1) circle (1); 
\fill [gray] (1,0) circle (1); \endscope
\scope \clip (0,0) circle (1);
\fill [gray] (0.5,-1) circle (1); \endscope;
\draw (0,0) circle (1) (0,1)  node [text=black] {$v(p)$}       (1,0) circle (1) (1,1)  node [text=black] {$v(q)$};
\draw (0.5,-1) circle (1) node [text=black,below=0.5cm] {$v(r)$};
\end{scope} 

\begin{scope}[xshift=2cm,scale=0.5,yshift=-8cm] 
\scope \clip (0,0) circle (1);
\fill [gray] (0.5,-1) circle (1); \endscope;
\draw (0,0) circle (1) (0,1)  node [text=black] {};
\draw (1,0) circle (1);
\draw (0.5,-1) circle (1) node (a) [text=black,below=0.5cm] {$v(p\land r)$};
\end{scope} 
\begin{scope}[xshift=5cm,scale=0.5,yshift=-8cm] 
\scope \clip (-2,-2)   (0.5,-1) circle (1); 
\fill [gray] (1,0) circle (1); \endscope
\draw (0,0) circle (1) (0,1)  node [text=black] {}       (1,0) circle (1) (1,1)  node [text=black] {};
\draw (0.5,-1) circle (1) node (b) [text=black,below=0.5cm] {$v(q\land r)$};
\end{scope} 
\begin{scope}[xshift=8cm,scale=0.5,yshift=-8cm] 
\fill [gray] (1,0) circle (1);
\fill [white] (0,0) circle (1);
\draw (0,0) circle (1) (0,1)  node [text=black] {}       (1,0) circle (1) (1,1)  node [text=black] {};
\draw (0.5,-1) circle (1) node (c) [text=black,below=0.5cm] {$v(\neg p \land q)$};
\draw (a) +(0,-1) node {2};
\draw (b) +(0,-1) node {3};
\draw (c) +(0,-1) node {4};

\end{scope} 
\end{tikzpicture}
}\hfil}

We can also compare this method to the traditional truth tables, showing that the results are the same from two different approaches:

\begin{center}
  {\vspace{0.5cm}}\begin{tabular}{|c|c|c|c|c|} 
    \hline
    $p$ & $q$ & $r$ & $(p\lor q)\land(\neg p\lor r)$ & States \\
     \hline
     \rowcolor{gray!50} 1 & 1 & 1 & 1 & 2,3 \\    \hline
    1 & 1 & 0 & 0 & \\    \hline
      \rowcolor{gray!50} 1 & 0 & 1 & 1 & 2 \\   \hline
    1 & 0 & 0 & 0 &  \\    \hline
     \rowcolor{gray!50} 0 & 1 & 1 & 1 & 3,4 \\   \hline
      \rowcolor{gray!50} 0 & 1 & 0 & 1 & 3  \\  \hline
    0 & 0 & 1 & 0 &  \\    \hline
    0 & 0 & 0 & 0 & \\
    \hline
  \end{tabular}
\end{center}

The truth table offers the complete disjunctive normal form of a wff, since each line that satisfies the wff corresponds to a conjunctive clause, and the whole table can be seen as the disjunction of these (first formula). We can compare it to the disjunctive normal form offered by the semantic tableau:

\vskip 1.5mm
{\raggedright $ (p\lor q)\land(\neg p\lor r) \equiv (p\land q \land r)\lor (p\land \neg q \land r)\lor (\neg p \land q\land r)\lor (\neg p\land q \land \neg r)$}

{\raggedright $ (p\lor q)\land(\neg p\lor r) \equiv (p \land r) \lor (\neg p \land q)  \lor (q \land r)$}
\vskip 1.5mm

The last column of the table show the equivalence between each truth values for the atoms and their corresponding states. For example, state 2 satisfies the first and the third line, since $\{2\}$ is a model for $(p\land q \land r)$ and $(p\land \neg q \land r)$. This comparison also shows that the use of semantic tableaux is usually faster than the construction of truth tables, being this a remarkable advantage of semantic tableaux for checking satisfiability over truth tables in computational terms.

\subsection{Validity}
If a formula $\varphi$ is valid, the tableau that begins with $\varphi$ does not give enough information to prove validity. A valid formula is also satisfiable, so the semantic tree of a valid formula offers an open tree determining a model for it but, when the formula is valid, one or several models are not enough, since all models satisfy the formula. However, we can determine the validity of $\varphi$ checking the satisfiability of its negation ($\varphi$ is valid iff the semantic tableau of $\neg \varphi$ is has all of its branches closed), and this is also
easily understandable in set-theoretical terms: if $\neg\varphi$ has no propositional model that satisfy it, all propositional models satisfy $\varphi$, so it is a valid formula.
 Therefore, the direct set-theoretical interpretation of semantic tableaux is not useful for checking the validity of wff.

Let's see an example of the cumbersomeness of the direct approach of this method. If we check a valid formula with this method, for example $(p \land q) \rightarrow (p\lor q)$, we can see that we obtain the following tree:

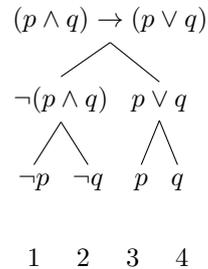
\begin{center}
\begin{tikzpicture}
\Tree [.$(p\land q)\rightarrow (p\lor q)$ [.$\neg (p\land q)$ [.\node(r1){$\neg p$};  ] [.$\neg q$ ]  ] [.$p\lor q$ [.$p$ ] [.$q$ ] ] ]

\node (1) [below of = r1,node distance=1cm] {1};
\node (2) [right of = 1,node distance=0.65cm] {2};
\node (3) [right of = 2,node distance=0.65cm] {3};
\node (4) [right of = 3,node distance=0.65cm] {4};

\end{tikzpicture}
\end{center}

It is obvious that the formula is satisfiable, and that all of the models of the branches are non-empty, $v(p)=\{3\}$, $v(q)=\{4\}$, $v(\neg p)=\{1\}$, $v(\neg p)=\{2\}$, and $v(\varphi )=\{1,2,3,4\}$, but it is not immediate nor evident whether the set of states that satisfy the formula equals to the whole universe of states $\mathcal{U}$, which is the condition for being a valid model. We have to test this with an ulterior method. In this case, it is not difficult to see that

\[\{1\} \cup \{2\} \cup \{3\} \cup \{4\} = v(p)\cup v(\neg p)\cup v(q)\cup v(\neg q) = \mathcal{U} \]
because we have two pair of complementary literals, $v(p)\cup v(\neg p)= \mathcal{U}$ and  $v(q)\cup v(\neg q)= \mathcal{U}$. We could have done this directly, ommiting the semantic tableau, so its use is neither useful nor practical. In sum, the semantic tableau of $\varphi$ only shows if it is satisfiable or not; if we want to check validity using semantic tableaux, we must check the satisfiability of $\neg \varphi$ and proceed as stated.

\section {Concluding remarks}
This proposal stresses semantics in terms of set theory and the use of semantic tableaux for teaching propositional logic. Both methods together, supplemented with other conceptual tools and showing their complementarieness, can make logic easier to students of any orientation and origin (humanities, computer sciences, etc.). So, for example, it is faster than truth tables, and it can be shown the relationship between a set-theoretical interpretation and a truth table.

As we have said, several advantages of using and combining these techniques can be considered. First, a basic set theory may be known for students of computer sciences or mathematics, but we propose to put it in action to reinforce its notions and appliactions. This basic theory provides a thorough way of working to students of humanities, which can lead gradual access to formal treatment of certain issues and promote the taste for the rigour of thinking. Of course, this facilitates the approach (and understanding) to first order semantics and propositional modal logic in general. Likewise, propisitonal models in terms os set theory of formulas can be obtained systematically by means of tableaux: when a formula is satisfiable, the corresponding tableau allows us to define a minimal universe of states and values for atoms.

Second, tableaux method can be taught as a set of rules of inference based on a semantic compositionality principle, an intutive method of analysis of inferences. On the other hand, for all kinds of students these methods can develop transversal logical capacities and compentences, which can help them to work other fields (as graph theory, algorithms, language analysis, etc.). In order to normalization of formulas, by means of tableaux disjunctive normal forms can be defined, then the path of learning resolution methos may be open.

To finish, we would like underlying that, though we have no statistical analysis of results, we have some experience in practicing our didactical points of view in Univesity of Seville, particularly with students of humanities, whose understanding of main logical issues has imporved during the last academic years.


\end{document}